\documentclass[12pt,reqno]{amsart}
\usepackage{graphicx}
\usepackage{amssymb,amsmath}
\usepackage{amsthm}
\usepackage{color,graphicx}
\usepackage{hyperref}
\usepackage{color}
\usepackage{epstopdf}
\usepackage{lpic}
\usepackage{cancel}
\usepackage{relsize}

\newcommand{\be}{\begin{equation}}
\newcommand{\ee}{\end{equation}}

\newtheorem{theorem}{Theorem}
\newtheorem{proposition}{Proposition}
\newtheorem{remark}{Remark}
\newtheorem{lemma}{Lemma}
\newtheorem{corollary}{Corollary}
\newtheorem{definition}{Definition}

\newcommand{\eq}[2]{\begin{equation}\begin{split}#1\end{split}\label{#2}\end{equation}}

\begin{document}

\title[Smooth solitary waves in the $b$-Camassa-Holm equation]{Stability of smooth solitary waves in the $b$-Camassa--Holm equation}

\author[S. Lafortune]{St\'ephane Lafortune}
\address[S. Lafortune]{Department of Mathematics, College of Charleston, Charleston, SC 29401, USA}
\email{lafortunes@cofc.edu}

\author{Dmitry E. Pelinovsky}
\address[D.E. Pelinovsky]{Department of Mathematics and Statistics, McMaster University,	Hamilton, Ontario, Canada, L8S 4K1}
\email{dmpeli@math.mcmaster.ca}

\date{\today}
\maketitle

\begin{abstract} 
	We derive the precise stability criterion for smooth solitary waves in the $b$-family of Camassa--Holm equations. 
	The smooth solitary waves exist on the constant background. In the integrable cases $b = 2$ and $b = 3$, we show analytically 
	that the stability criterion is satisfied and smooth solitary waves are orbitally stable with respect to perturbations in $H^3(\mathbb{R})$.  In the non-integrable cases, we show numerically and asymptotically that the stability criterion is satisfied for every $b > 1$. The orbital stability theory relies on a different Hamiltonian formulation compared to the Hamiltonian formulations available in the integrable cases.
\end{abstract} 

\section{Introduction}

The $b$-family of Camassa--Holm equations (which we simply call $b$-CH) is written for the scalar velocity variable $u = u(t,x)$ in the form
\begin{equation}
\label{bCH}
u_t-u_{txx}+(b+1)uu_x = bu_xu_{xx}+uu_{xxx},
\end{equation}
where $b$ is arbitrary parameter. The $b$-CH model was introduced in \cite{dhh,Dullin} by using transformations of the integrable hierarchy of KdV equations. The $b$-CH model is not integrable in general but it has the same asymptotic accuracy as the integrable cases of $b = 2$ called the Camassa--Holm equation \cite{Cam} and $b = 3$ called the Degasperis--Procesi equation  \cite{dp}. The $b$-CH model describes the horizontal velocity $u = u(t,x)$ for the unidirectional propagation of waves along the surface of a shallow water flowing over a flat bed at a certain depth \cite{rossen}. The hydrodynamical relevance of the Camassa--Holm and Degasperis--Procesi equations for modelling of shallow water waves was discussed in 
\cite{Cam2,Const4,Johnson}. 
 
Peaked and smooth solitary waves exist in the $b$-CH equation (\ref{bCH}), depending on the values of the parameter $b$ and the parameter $k$ for the constant background. Traveling waves in the $b$-CH equation were studied by 
using dynamical system methods \cite{Guo} and hodograph transformations \cite{Hone2022}.

Early numerical simulations on the zero background ($k=0$) in \cite{Holm1,Holm2} showed that the initial data resolves into a sequence of peaked solitary waves called {\em peakons} for $b > 1$ and a sequence of smooth solitary waves 
called {\em leftons} for $b < -1$. In the intermediate case of $b \in (-1,1)$, 
the initial data generates a rarefactive wave with exponentially decaying tails.
Recent numerical experiments in \cite{Char1} added more examples of dynamics of peaked solitary waves which appear to be unstable for $b < 1$ and stable for $b > 1$.
 
Stability of both peaked and smooth solitary waves on the zero background  
($k=0$) has been now well understood. Orbital stability of {\em leftons} for $b < -1$ 
was shown in \cite{Hone14} by using the variational formulation from \cite{dhh-proc} and analyzing perturbations in some 
exponentially weighted spaces. Orbital stability of {\em peakons} was shown 
for $b = 2$ in \cite{Const5,Cons1} and for $b = 3$ in \cite{LinLiu} by using conservation of two energy integrals in the energy space 
$H^1(\mathbb{R})$ for $b = 2$ and $L^2(\mathbb{R}) \cap L^3(\mathbb{R})$ for $b = 3$. However, the initial-value problem for the $b$-CH equation with $b > 1$ is ill-posed in $H^s(\mathbb{R})$ for $s < \frac{3}{2}$ \cite{Him} due to the lack of continuous dependence and norm inflation, hence the orbital stability 
in the energy space is only conditional with respect to the existence 
of local solutions. It was shown in \cite{Linares} that $H^1(\mathbb{R}) \cap W^{1,\infty}(\mathbb{R})$ is the largest space 
where the initial-value problem near the peaked solitary waves is defined. 
However, the $W^{1,\infty}(\mathbb{R})$ norm of perturbations 
grows generally and induces the nonlinear instability of {\em peakons} 
in $H^1(\mathbb{R}) \cap W^{1,\infty}(\mathbb{R})$. This was first shown 
with the method of characteristics for $b = 2$ in \cite{Natali} (and in the periodic settting, in \cite{MP-2021}) and for the cubic Novikov equation 
in \cite{ChenPelin}. In our previous work \cite{LP-21}, we have proven spectral and linear instability 
of {\em peakons} in the $b$-CH equation for every $b$.

Solitary waves on the constant background $k \neq 0$ are smooth. The smooth solitary waves 
can be found from the transformation $u(t,x) = k + v(t,x-kt)$, where 
$v(t,x)$ satisfies the equivalent version of the $b$-CH equations:
\begin{equation}
\label{bCH-v}
v_t-v_{txx}+(b+1)v v_x = k v_x + b v_x v_{xx}+ v v_{xxx}.
\end{equation}
If $v(t,x) \to 0$ as $|x| \to \infty$, then $u(t,x) \to k$ as $|x| \to \infty$.

Orbital stability of smooth solitary waves was obtained for $b = 2$ in \cite{CS-02} and for $b = 3$ in \cite{Liu-21} by working with the conserved energy integrals in the energy space. Since the Hamiltonian structure used in \cite{CS-02} and \cite{Liu-21} is a special feature due to integrability of the $b$-CH equations for $b = 2$ and $b = 3$, these results cannot be extended to other values of $b$. 
 
{\em The purpose of this work is to study orbital stability of the smooth 
	solitary waves in the $b$-CH equations for any $b > 1$.} 
We explore 
the Hamiltonian formulation of the $b$-CH equation (\ref{bCH}) from \cite{dhh-proc} 
and the characterization of smooth traveling waves from \cite{GeyerSAPM}. 
As a result of relevant computations, we obtain the precise criterion for orbital stability of the smooth solitary waves with respect to perturbations in $H^3(\mathbb{R})$. 
We show analytically for $b = 2$ and $b = 3$ and numerically for other values of $b > 1$ that the stability criterion is satisfied. This yields orbital stability of the smooth solitary waves with respect to perturbations in $H^3(\mathbb{R})$.

The Hamiltonian formulation of the $b$-CH equation (\ref{bCH}) is developed 
by using the momentum density $m := u - u_{xx}$, for which the $b$-CH equation (\ref{bCH}) can be rewritten in the form:
\begin{equation}\label{bCHm}
m_t + u m_x + b m u_x = 0.
\end{equation}
If $b \neq 1$, this equation can be cast in the Hamiltonian form
\begin{equation}
\label{sympl-1}
\frac{dm}{dt} = J_m \frac{\delta E}{\delta m}, 
\end{equation}
where 
$$
J_m := -\frac{1}{b-1} (bm \partial_x + m_x) (1-\partial_x^2)^{-1} \partial_x^{-1} (b \partial_x m - m_x)
$$
is the skew-adjoint operator in $L^2(\mathbb{R})$ such that $J_m^* = -J_m$ and 
\begin{equation}
\label{E}
E(m) = \int_{\mathbb{R}} m dx
\end{equation} 
is the conserved mass integral. When solutions of the $b$-CH equation (\ref{bCHm}) are considered on the zero background, there exist two other conserved quantities \cite{dhh-proc} given by 
\begin{equation}
\label{Eu}
F_1(m) = \int_{\mathbb{R}} m^{\frac{1}{b}} dx
\end{equation}
and
\begin{equation}
\label{Fu}
F_2(m) = \int_{\mathbb{R}} \left( \frac{m_x^2}{b^2 m^2} + 1 \right) m^{-\frac{1}{b}} dx.
\end{equation}
The conserved quantities $E$ and $F_2$ were used in the study 
of orbital stability of {\em leftons} for $b < -1$ 
in the exponentially decaying spaces, for which these two functionals are well-defined \cite{Hone14}.

When solutions of the $b$-CH equation (\ref{bCHm}) are considered on the nonzero constant background with $m(t,x) \to k$ decaying fast as $|x| \to \infty$, we have to redefine 
the conserved quantities $F_1(m)$ and $F_2(m)$ in order to eliminate divergence of the constant background. For $k > 0$, we consider the class of functions in the set 
\begin{equation}
\label{function-space}
X_k = \left\{ m - k \in H^1(\mathbb{R}) : \quad m(x) > 0, \;\; x \in \mathbb{R} \right\}
\end{equation}
and redefine the conserved quantities for $m \in X_k$ as
\begin{equation}
\label{mass-hat}
\hat{E}(m) = \int_{\mathbb{R}} (m-k) dx,
\end{equation}
\begin{equation}
\label{Eu-hat}
\hat{F}_1(m) = \int_{\mathbb{R}} \left[ m^{\frac{1}{b}} - k^{\frac{1}{b}} \right] dx
\end{equation}
and
\begin{equation}
\label{Fu-hat}
\hat{F}_2(m) = \int_{\mathbb{R}} 
\left[ \left( \frac{m_x^2}{b^2 m^2} + 1 \right) m^{-\frac{1}{b}} 
- k^{-\frac{1}{b}}  \right] dx.
\end{equation}
Since $m(x) \to k$ as $|x| \to \infty$, $k > 0$, and $m(x) > 0$, 
there exists $m_0 > 0$ such that $m(x) \geq m_0$ for every $x \in \mathbb{R}$. 
Therefore, $\hat{F}_1(m)$ and $\hat{F}_2(m)$ are well-defined in $X_k$ due to 
the Banach algebra property of $H^1(\mathbb{R})$ and the conservation 
of $\hat{E}(m)$. The proof of orbital stability 
of smooth solitary waves becomes simpler compared to \cite{Hone14}.

Local well-posedness of the $b$-CH equation (\ref{bCH-v}) with $k > 0$ can be easily shown for the initial data $v(0,\cdot) = v_0 \in H^s(\mathbb{R})$ with $s > \frac{3}{2}$ \cite{ConstEcher,Escher,Zhou}. 
This well-posedness theory for the function $v \in H^s(\mathbb{R})$ translates to the function $m \in X_k$ if $s = 3$. Moreover, it is well-known that the local solutions exist for all times 
and do not break if the momentum density $m(t,x)$ is strictly positive \cite{C-E,Escher,Zhou}, which is true for $m \in X_k$. Assuming that the local and global well-posedness of the $b$-CH equation (\ref{bCH-v}) is well-known, we adopt the following definition of orbital stability of travelling waves in the set $X_k$. 

\begin{definition}
	\label{def-main}
	Let $m(t,x) = \mu(x-ct)$ be the travelling wave solution of the $b$-CH equation (\ref{bCHm}) with $\mu \in X_k$. We say that the travelling wave is orbitally stable in $X_k$ if for every $\varepsilon > 0$ there exists $\delta > 0$ such that for every $m_0 \in X_k$ satisfying 
	$\| m_0 - \mu \|_{H^1} < \delta$, there exists a unique solution 
	$m \in C^0(\mathbb{R},X_k)$ of the $b$-CH equation (\ref{bCHm}) 
	with the initial datum $m(0,\cdot) = m_0$ satisfying 
	$$
	\inf_{x_0 \in \mathbb{R}} \| m(t,\cdot) - \mu(\cdot-x_0) \|_{H^1} < \varepsilon, \quad t \in \mathbb{R}.
	$$
\end{definition}

The following theorem presents the main result of this work. 

\begin{theorem}
	\label{theorem-main}
	For fixed $b > 1$, $c > 0$, and $k \in (0,(b+1)^{-1} c)$, there exists a unique solitary wave $m(t,x) = \mu(x-ct)$ of the $b$-CH equation (\ref{bCHm}) with profile $\mu \in C^{\infty}(\mathbb{R})$ satisfying $\mu(x) > 0$ for $x \in \mathbb{R}$, $\mu'(0) = 0$, and $\mu(x) \to k$ as $|x| \to \infty$ exponentially fast. The solitary wave is orbitally stable in $X_k$ if the mapping 
	\begin{equation}
	\label{map-intro}
	k \mapsto Q(\phi) := \int_{\mathbb{R}} 
	\left[ b \left( \frac{c-k}{c-\phi} \right) - \left(\frac{c-k}{c-\phi} \right)^{b} -b+1 \right]   dx
	\end{equation}
	is strictly increasing, where $\phi := k + (1- \partial_x^2)^{-1} (\mu - k)$ is uniquely defined. 
\end{theorem}

\begin{remark}
	Theorem \ref{theorem-main} implies that if $u(t,x) = \phi(x-ct)$ is the  travelling wave solution of the $b$-CH equation (\ref{bCH}) with $\phi \in Y_k$ given by 
\begin{equation}
\label{function-space-Y}
Y_k = \left\{ u - k \in H^3(\mathbb{R}) : \quad u(x) - u''(x) > 0, \;\; x \in \mathbb{R} \right\},
\end{equation}	
then it is orbitally stable in $Y_k$ in the $H^3(\mathbb{R})$ norm. Indeed, integration by parts yields
$$
\| m \|_{H^1}^2 = \int_{\mathbb{R}} (m^2 + m_x^2) dx 
= \int_{\mathbb{R}} (u^2 + 3 u_x^2 + 3 u_{xx}^2 + u_{xxx}^2) dx,
$$
which is equivalent to the squared $H^3(\mathbb{R})$ norm on $u$.
\end{remark}

\begin{remark}
	The stability criterion of 	Theorem \ref{theorem-main} is verified for $b = 2$ and $b = 3$ based on analytical computations. This gives an alternative proof of the orbital stability of smooth solitary waves in the Camassa--Holm and Degasperis--Procesi equations compared to \cite{CS-02} and \cite{Liu-21}, respectively.
\end{remark}

\begin{remark}
	The stability criterion of 	Theorem \ref{theorem-main} is only verified numerically for other values of $b > 1$ and asymptotically in the limit $k \to 0$ and $k \to (b+1)^{-1} c$. It is an open question on proving the stability criterion analytically in the general case.
\end{remark}

\begin{remark} 
	Additional families of smooth traveling solitary waves also exist in $X_k$ for $b \leq 1$ and their precise stability criterion can be obtained 
	similarly by using the Hamiltonian form (\ref{sympl-1}) with the same $J_m$ if $b < 1$ and with the modified expression for $J_m$ if $b = 1$ \cite{dhh-proc}. 
\end{remark}

\begin{remark}
	Orbital stability of smooth multi-soliton solutions of the Camassa--Holm equation with $b = 2$ was recently proven in \cite{LiuWang}. The proof relies on the bi--Hamiltonian structure of the Camassa--Holm equation and may not be generalized for a general case $b > 1$, for which the only Hamiltonian structure is given by (\ref{sympl-1}).
\end{remark}

The article is organized as follows. Travelling waves including the smooth solitary waves on the nonzero constant background are characterized in Section \ref{sec-trav}. Variational characterization of the travelling solitary waves 
in terms of the mass and energy integrals is developed in Section \ref{sec-var}. Derivation and proof of the stability criterion in Theorem \ref{theorem-main} are given in Section \ref{sec-stab}. 
Verification of the stability criterion is reported in Section \ref{sec-num}. 
We give the summary and discuss open directions in the concluding Section \ref{sec-conclude}.

\section{Traveling waves}
\label{sec-trav}

Let us consider traveling waves of the $b$-CH equation (\ref{bCH}) in the form $u(t,x) = \phi(x-ct)$ with speed $c$ and profile $\phi$ found from the third-order differential equation
\begin{equation}
\label{third-order}
-(c-\phi) (\phi''' - \phi') + b \phi' (\phi'' - \phi) = 0.
\end{equation}
The following lemma characterizes the family of solitary waves on the nonzero constant background parameterized by the arbitrary parameter $k > 0$. 

\begin{lemma}
	\label{lem-trav}
	For fixed $b > 1$ and $c > 0$, there exists a one-parameter 
	family of smooth solitary waves with profile $\phi \in C^{\infty}(\mathbb{R})$ satisfying $\phi'(0) = 0$ and 
	$\phi(x) \to k$ as $|x| \to \infty$ if and only if the arbitrary parameter $k$ belongs 
	to the interval $(0,(b+1)^{-1} c)$. Moreover, 
	\begin{equation}
	\label{smooth-soliton}
	0 <	\phi(x) < c, \qquad x \in \mathbb{R},
	\end{equation}
	and the family is smooth with respect to parameter $k$ in $(0,(b+1)^{-1} c)$.
\end{lemma}

\begin{proof}
Multiplying (\ref{third-order}) by $(c-\phi)^{b-1}$ and integrating in $x$ 
yield the second-order equation:
\begin{equation}
\label{second-order}
-(c - \phi)^b (\phi'' - \phi) = a,
\end{equation}
where $a$ is the integration constant. The second-order equation 
(\ref{second-order}) is conservative and admits the first-order quadrature:
\begin{equation}
\label{first-order}
\frac{1}{2} (b-1) (\phi'^2 - \phi^2) + \frac{a}{(c-\phi)^{b-1}} = g,
\end{equation}
where $g$ is another integration constant.

Smooth solitary wave solutions with profile $\phi \in C^{\infty}(\mathbb{R})$ satisfying $\phi(x) \to k$ as $|x| \to \infty$ correspond to the homoclinic 
orbit from the equilibrium point $(\phi,\phi') = (k,0)$. Taking the limit as $|x|\rightarrow \infty$ in \eqref{second-order} and \eqref{first-order}
yields the relations:
\begin{equation}
\label{ag}
	a = k (c-k)^b, \qquad g = kc - \frac{1}{2} (b+1)k^2.
\end{equation}

For fixed $b > 1$ and $c > 0$, the first-order quadrature (\ref{first-order}) 
represents the energy conservation for a Newtonian particle 
with the mass $m := b-1 > 0$ and energy $g$ under a force with 
the potential energy
$$
U(\phi) := - \frac{1}{2} (b-1) \phi^2 + \frac{a}{(c-\phi)^{b-1}}.
$$
For smooth solutions $\phi \in C^{\infty}(\mathbb{R})$, we consider the restriction $\phi \in (-\infty,c)$. It can be shown that no smooth solitary waves exist for $\phi \in (c,\infty)$ if $b > 1$.

Critical points of $U$ in $(-\infty,c)$ are given by roots of the algebraic equation $\phi (c-\phi)^b = a$ for $\phi \in (-\infty,c)$. Since the global maximum of $\phi \mapsto \phi (c - \phi)^b$ on $(-\infty,c)$
occurs at $\phi = \frac{c}{b+1}$, there exists only one critical point of $U$ for $a \in (-\infty,0]$ and $a = \mathfrak{a}$, two critical points of $U$  for $a \in (0,\mathfrak{a})$, and no critical points 
of $U$  for $a \in (\mathfrak{a},\infty)$, where 
\begin{equation}
\label{frak-a}
\mathfrak{a} := \frac{b^b c^{b+1}}{(b+1)^{b+1}}. 
\end{equation}
Homoclinic orbits with $\phi \in (-\infty,c)$ exist only if at least two roots of the algebraic equation $\phi (c-\phi)^b = a$ exist in $(-\infty,c)$, which 
happens if and only if $a \in (0,\mathfrak{a})$. The two roots 
can be ordered as follows:
\begin{equation}
\label{ordering}
0 < \phi_1 < \frac{c}{b+1} < \phi_2 < c.
\end{equation}
We will now show that the homoclinic orbit does exist if $a \in (0,\mathfrak{a})$.

For fixed $b > 1$ and $c > 0$, the local maximum and minimum points of $U$ gives respectively the saddle point $(\phi_1,0)$ and the center point $(\phi_2,0)$ of the second-order equation (\ref{second-order}). There exists a punctured neighbourhood of the center $(\phi_2,0)$ enclosed by the homoclinic orbit connecting the saddle $(\phi_1,0)$. Thus, $\phi_1 \equiv k \in (0,(b+1)^{-1}c)$ is taken as the arbitrary parameter of the homoclinic orbit, 
which specifies $a$ and $g = U(k)$ by the relation (\ref{ag}). Since $U(\phi) \to +\infty$ as $\phi \to c$ from the left, the homoclinic orbit belongs to the vertical stripe $\{ (\phi,\phi') : 0 < \phi < c \}$ and represents the smooth solitary wave with the profile $\phi \in C^{\infty}(\mathbb{R})$ satisfying  $\phi(x) \to k$ as $|x| \to \infty$. By the translational invariance, the solitary wave profile satisfying $\phi'(0) = 0$ is uniquely defined. 

Finally, smoothness of the family in $k$ is due to the fact that the differential equations depend smoothly on $\phi$, $g$, $a$, and $c$ 
if the solution belongs to the range in (\ref{smooth-soliton}), whereas parameters $a$ and $g$ depends smoothly on $k$ in (\ref{ag}) if $k \in (0,(b+1)^{-1} c)$.
\end{proof}

\begin{figure}[htb!]
	\includegraphics[width=0.48\textwidth]{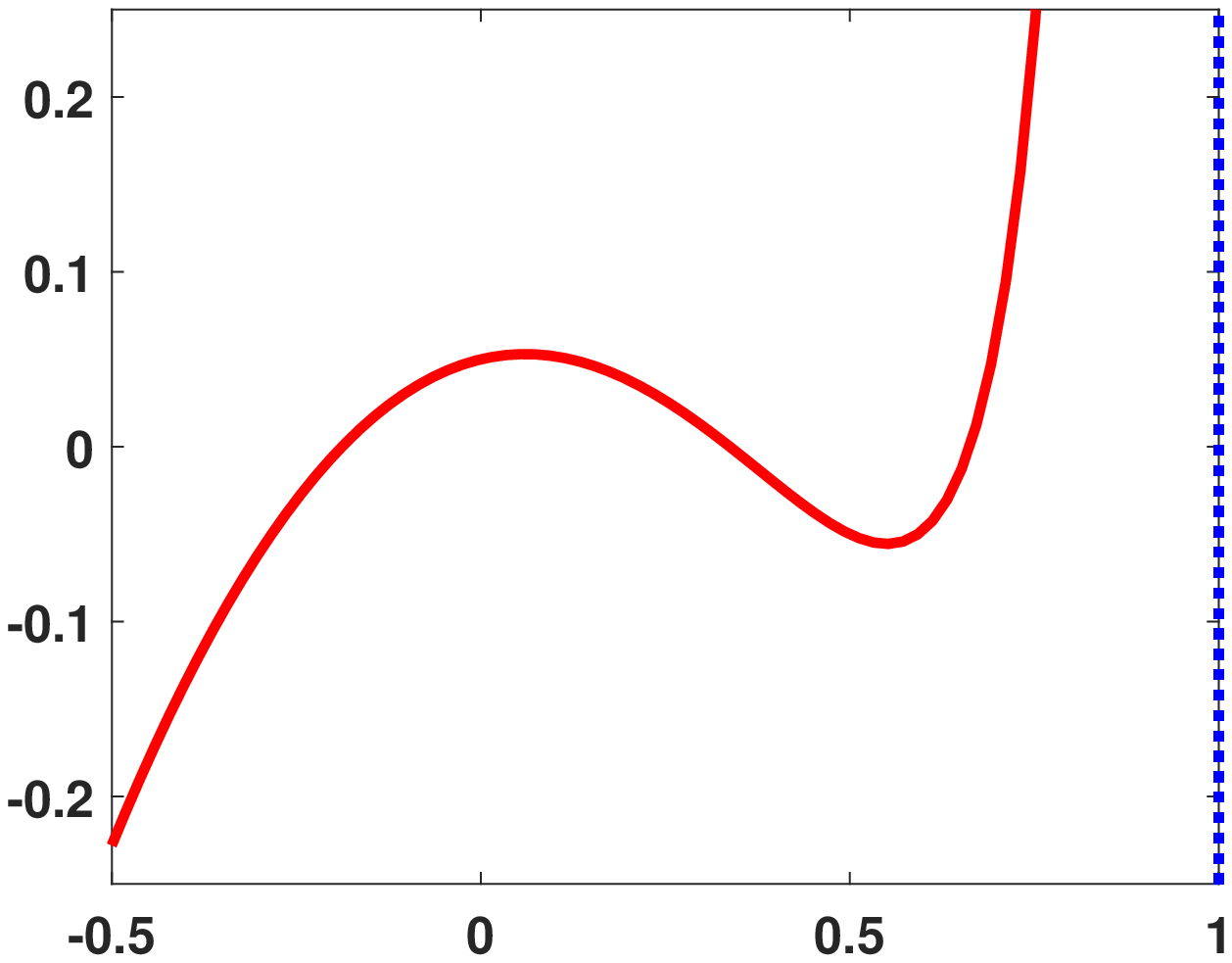}
	\includegraphics[width=0.48\textwidth]{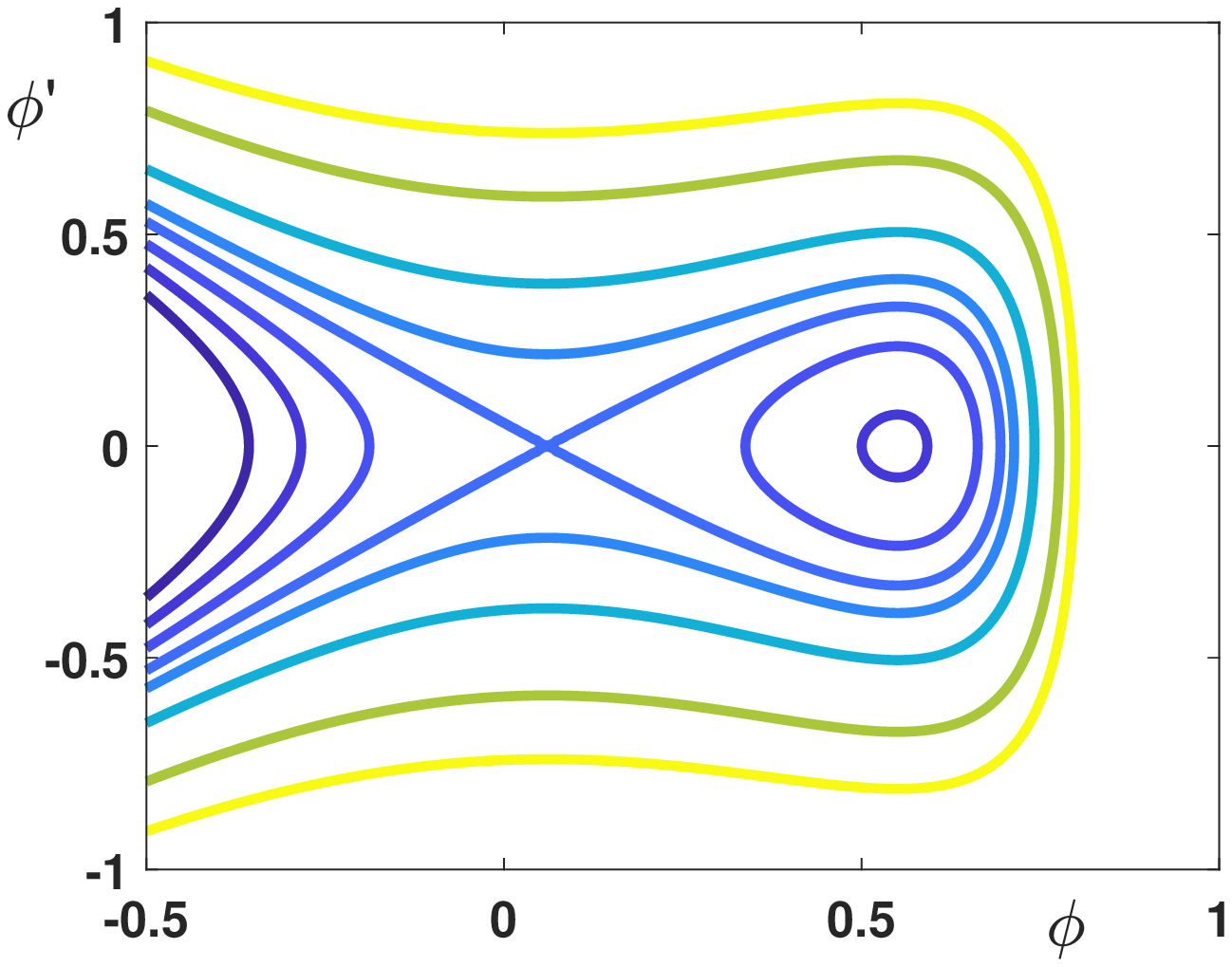}
	\caption{Left: $U$ versus $\phi$ for $b = 3$, $c = 1$, and $a = 0.05 \in (0,\mathfrak{a})$. Right: the phase portrait of the second-order equation 
		(\ref{second-order}) constructed from the level curves of the first-order invariant (\ref{first-order}) on the phase plane $(\phi,\phi')$ for the same parameter values.} \label{fig-plane}
\end{figure}

For illustration of the proof of Lemma \ref{lem-trav}, Figure \ref{fig-plane} (left)  shows the graph of $U$ versus $\phi$ for $b = 3$, $c = 1$, and $a = 0.05 \in (0,\mathfrak{a})$, where $\mathfrak{a}\approx 0.11$.  
The corresponding phase portrait on the phase plane $(\phi,\phi')$  is shown 
on the right panel.

\begin{remark}
	\label{remark-standard}
	The standard integration of the third-order equation (\ref{third-order}) in $x$ gives a different second-order equation
	\begin{equation}\label{CHode}
	-(c-\phi)(\phi''-\phi) + \frac{1}{2} (b-1) (\phi'^2 - \phi^2) = g,
	\end{equation}
which can be integrated to the same first-order quadrature (\ref{first-order}). 
Consequently, the second-order equation (\ref{CHode}) is redundant in view 
of the two equations (\ref{second-order}) and (\ref{first-order}).
\end{remark}

\begin{remark}		\label{remark-scaling}
	The one-parameter family of smooth solitary waves exists for fixed $b > 1$ and $c < 0$. It is obtained from the family in Lemma \ref{lem-trav} by using the symmetry transformation:
\begin{equation}
		\label{scal-transform}
c \mapsto -c, \quad	\phi \mapsto -\phi, \quad 
a \mapsto (-1)^{b+1} a, \quad g \mapsto g,
\end{equation}
which leaves the system of equations (\ref{second-order}) and (\ref{first-order}) invariant. For $b > 1$ and $c < 0$, the arbitrary parameter $k$ belongs to the interval $((b+1)^{-1}c,0)$. This family satisfies $\mu(x) < 0$ for all $x \in \mathbb{R}$ so it does not belong to $X_k$ in (\ref{function-space}).
\end{remark}

\begin{remark}
	\label{remark-limit}
	In the limiting case $k \to 0$, the profile $\phi$ is no longer smooth since $a \to 0$ for which $\phi_1 = 0$ and $\phi_2 = c$. This limit recovers the peaked solitary waves 
	with the profile $\phi(x) = c e^{-|x|}$ considered in our previous work 
	\cite{LP-21}. In the limiting case $k \to (b+1)^{-1} c$, 
	the profile $\phi$ is constant in $x$ since $a \to \mathfrak{a}$ 
	for which $\phi_1 = \phi_2 = (b+1)^{-1} c$. 
\end{remark}

\begin{remark}
	\label{remark-momentum}
If we denote the momentum density for the solitary wave with profile $\phi$ by $\mu := \phi - \phi''$, then the second-order equation \eqref{second-order} 
	gives the relation:
\begin{equation}
\label{muphi}
\mu=\frac{a}{(c-\phi)^b}.
\end{equation}
The relation (\ref{muphi}) can be recovered with the traveling wave reduction 
$m(t,x) = \mu(x-ct)$ of the $b$-CH equation in the momentum form (\ref{bCHm}), from which $\mu$ satisfies the differential equation
	\begin{equation}
	\label{m-phi-connection}
	\mu'(\phi - c) + b \mu \phi' = 0.
	\end{equation}
	After multiplying (\ref{m-phi-connection}) by $(c-\phi)^{b-1}$ and integration in $x$, we obtain (\ref{muphi}).
\end{remark}

\begin{remark}
	\label{remark-positive}
	Inequalities (\ref{smooth-soliton}) for the smooth solitary waves of Lemma \ref{lem-trav} are {\em strict} in the sense that for every fixed $b > 1$, $c > 0$, and $k \in (0,(b+1)^{-1} c)$, there exist positive 
constants $C_{\pm}(b,c,k)$ such that 
		\begin{equation}
	\label{smooth-soliton-strict}
C_-(b,c,k) \leq c -	\phi(x) \leq C_+(b,c,k), \qquad x \in \mathbb{R}.
	\end{equation}
Consequently, there exist positive constants $\hat{C}_{\pm}(b,c,k)$ such that 
\begin{equation}
\label{smooth-soliton-strict-mu}
\hat{C}_-(b,c,k) \leq 	\mu(x) \leq \hat{C}_+(b,c,k), \qquad x \in \mathbb{R}.
\end{equation}
In view of Remark \ref{remark-limit}, $C_-(b,c,k) \to 0$ and $\hat{C}_+(b,c,k) \to \infty$ as $k \to 0$ when the smooth solitary wave becomes the peaked solitary wave for every $b > 1$ and $c > 0$. 
\end{remark}

\begin{remark}
	By a similar phase plane analysis as in the proof of Lemma \ref{lem-trav}, 
	we can identify additional families of smooth solitary waves in $X_k$ 
	for $b \leq 1$. Their stability analysis can be developed 
	similarly to the proof of Theorem \ref{theorem-main} in the case $b > 1$. This is left to interested students as exercises.
\end{remark}

\section{Variational characterization}
\label{sec-var}

The solitary waves of Lemma \ref{lem-trav} can be characterized variationally if the system of equations (\ref{second-order}) and (\ref{first-order}) can be shown to arise as the Euler--Lagrange equation of the action functional 
\begin{equation}
\label{action-1}
\Lambda_{\omega_1,\omega_2}(m) := \hat{E}(m) - \omega_1 \hat{F}_1(m) - \omega_2 \hat{F}_2(m),
\end{equation}
where the normalized mass $\hat{E}(m)$ and the two energies $\hat{F}_1(m)$ and $\hat{F}_2(m)$ are given by (\ref{mass-hat}), (\ref{Eu-hat}), and (\ref{Fu-hat}). The following lemma states that the variational characterization is possible if and only if the Lagrange multipliers $\omega_1$ and $\omega_2$ are uniquely related 
to parameters of the differential equations 
(\ref{second-order}) and (\ref{first-order}).

\begin{lemma}
	\label{lem-EL}
For fixed $b > 1$ and $c > 0$, a critical point $\mu \in X_k$ of the action functional $\Lambda_{\omega_1,\omega_2}$ in (\ref{action-1}) coincides with 
	the solitary wave solution $\mu = \phi - \phi'' \in C^{\infty}(\mathbb{R})$ of Lemma \ref{lem-trav} satisfying (\ref{second-order}) and (\ref{first-order}) if and only if 
\begin{equation}
\label{multipliers-relation}
\omega_1 = \frac{1}{2} \left[ (b-1) c + (b+1) k \right] k^{-\frac{1}{b}}, 
\qquad \omega_2 = \frac{1}{2} (b-1) (c-k) k^{\frac{1}{b}},
\end{equation}
where $k \in (0,(b+1)^{-1} c)$ is an arbitrary parameter.
\end{lemma}

\begin{proof}
Let $\mu \in X_k$ be a critical point of $\Lambda_{\omega_1,\omega_2}$. After  straightforward simplifications, the equation $\Lambda_{\omega_1,\omega_2}'(\mu) = 0$ gives the following differential equation,
\begin{equation}
\label{EL-1}
1 - \frac{\omega_1}{b} \mu^{\frac{1}{b}-1} 
+ \frac{\omega_2}{b \mu^{\frac{1}{b}+1}} \left[ \frac{2\mu''}{b \mu} 
- \frac{(2b+1)  (\mu')^2}{b^2 \mu^2} + 1 \right] = 0.
\end{equation}
If $\mu \in X_k$ is a weak solution of the differential equation (\ref{EL-1}), then $\mu \in C^{\infty}(\mathbb{R})$ by bootstrapping arguments, so that 
the set of solitary wave solutions is given by Lemma \ref{lem-trav} if 
we can establish equivalence between the differential equations.

Assuming the relation (\ref{muphi}), we introduce $\phi$ such that 
$\mu (c - \phi)^b = a$ and obtain 
$$
\mu' = \frac{b \phi'}{c-\phi} \mu, \qquad 
\mu'' = \frac{b \phi''}{c-\phi} \mu + \frac{b(b+1) (\phi')^2}{(c-\phi)^2} \mu.
$$
Substituting these relations into (\ref{EL-1}) gives after straightforward simplifications:
\begin{equation}
\label{EL-2}
a b (c-\phi)^{1-b} - \omega_1 a^{\frac{1}{b}}
+ \omega_2 a^{-\frac{1}{b}} \left[ 2 (c-\phi) \phi'' 
- (\phi')^2 + (c-\phi)^2 \right] = 0.
\end{equation}
Substituting (\ref{CHode}) into (\ref{EL-2}) yields 
\begin{equation}
\label{EL-3}
ab (c-\phi)^{1-b} - \omega_1 a^{\frac{1}{b}}
+ \omega_2 a^{-\frac{1}{b}} \left[ b (\phi')^2 - b \phi^2 + c^2 - 2 g \right] = 0.
\end{equation}
Substituting (\ref{first-order}) into (\ref{EL-3}) gives the unique choice for Lagrange multipliers, 
\begin{equation}
\label{multipliers}
\omega_1 = \frac{1}{2 a^{\frac{1}{b}}} \left[ 2g + (b-1) c^2 \right], 
\quad \omega_2 = \frac{1}{2} a^{\frac{1}{b}} (b-1).
\end{equation}
Finally, substituting (\ref{ag}) into (\ref{multipliers}) gives (\ref{multipliers-relation}). 
\end{proof}

\begin{remark}
	\label{remark-var}
	Only one parameter $k \in (0,(b+1)^{-1}c)$ is arbitrary for fixed $b > 1$ and $c > 0$. Although the action functional (\ref{action-1}) has two Lagrange multipliers, first variations of $\hat{E}(m)$, $\hat{F}_1(m)$, and $\hat{F}_2(m)$ are not defined independently of each other for $m \in X_k$ due to nonzero boundary conditions at infinity.
\end{remark}

The variational characterization of Lemma \ref{lem-EL} implies two important 
properties when we add a perturbation $\tilde{m} := m - \mu$ to the solitary wave with the profile $\mu \in C^{\infty}(\mathbb{R})$.
Since $\mu(x)$ is strictly positive and bounded by Remark \ref{remark-positive}, we have $m \in X_k$ if 
$\tilde{m} \in H^1(\mathbb{R})$ and the $H^1(\mathbb{R})$ norm of $\tilde{m}$ is sufficiently small. 

The following two results describe the second-order 
variation of the action functional $\Lambda_{\omega_1,\omega_2}$ 
and the first-order constraint on the perturbation $\tilde{m}\in H^1(\mathbb{R})$. In what follows, $\langle \cdot,\cdot \rangle$ denotes the standard inner product in $L^2(\mathbb{R})$, $\mu := \phi - \phi'' \in C^{\infty}(\mathbb{R})$ is defined by  
Lemma \ref{lem-trav}, and $(\omega_1,\omega_2)$ are defined by Lemma \ref{lem-EL}. 

\begin{corollary}
	\label{cor-second-var}
There exists a sufficiently small positive $\epsilon_0$ such that 
	for every $\tilde{m}\in H^1(\mathbb{R})$ satisfying $\| \tilde{m} \|_{H^1} \leq \varepsilon_0$, we have 
\begin{equation}
\label{expansion-lambda}
\Lambda_{\omega_1,\omega_2}(\mu + \tilde{m}) - \Lambda_{\omega_1,\omega_2}(\mu) = \langle \mathcal{L}\tilde{m},\tilde{m} \rangle + R(\tilde{m}),
\end{equation}	
where 
\begin{eqnarray}
\nonumber
\mathcal{L} &=& -\frac{d}{dx} \frac{c-\phi}{\mu^2} \frac{d}{dx} 
+ \frac{(b+1) (c-\phi)}{2 \mu^2} + \frac{(2b+1) (c-\phi)\mu''}{b \mu^3} \\
&& - \frac{(2b+1)(3b+1) (c-\phi) (\mu')^2}{2 b^2 \mu^4} 
- \frac{(c-k)[(b-1)c + (b+1)k]}{2 (c-\phi) \mu^2}
\label{second-var-1}
\end{eqnarray} 
and $R(\tilde{m})$ is the remainder term satisfying 
$\| R(\tilde{m}) \|_{H^1} \leq C_0  \|\tilde{m}\|_{H^1}^3$ for some $\tilde{m}$-independent positive constant $C_0$.
\end{corollary}

\begin{proof}
The expression for $\mathcal{L}$ is obtained by straightforward computations with the use of relations (\ref{muphi}) and (\ref{multipliers-relation}). Coefficients of the Sturm--Liouville operator $\mathcal{L}$ are smooth and bounded since $(c-\phi), \mu$ are strictly positive, bounded, and smooth on $\mathbb{R}$ by Lemma \ref{lem-trav} and Remark \ref{remark-positive}. 

Similarly, $R(\tilde{m})$ is computed by using the Taylor expansion of the energy densities in $\hat{F}_1(m)$ and $\hat{F}_2(m)$ at $m = \mu$. The leading-order term in $R(\tilde{m})$ is cubic and the Sobolev space $H^1(\mathbb{R})$ forms a Banach algebra with respect to multiplication so that the estimate $\| R(\tilde{m}) \|_{H^1} \leq C_0  \|\tilde{m}\|_{H^1}^3$ follows from the Taylor expansion of $\hat{F}_1(m)$, $\hat{F}_2(m)$ and the smallness of  $\| \tilde{m} \|_{H^1}$. 
\end{proof}

\begin{corollary}
	\label{cor-constraint}
	Let $\tilde{m} \in H^1(\mathbb{R})$ be a small perturbation to $\mu$ such that $m = \mu + \tilde{m} \in X_k$ does not change the conserved quantity 
	$\hat{F}(m) := b \hat{F}_1(m)- k^{\frac{1}{b}-1} \hat{E}(m)$ up to the first order. Then, $\tilde{m} \in H^1(\mathbb{R})$ satisfies the constraint 
\begin{equation}
\label{equn}
\langle \mu^{\frac{1}{b}-1}-{k^{\frac{1}{b}-1}}, \tilde{m} \rangle = 0.
\end{equation}
\end{corollary}

\begin{proof}
	Since $\hat{F}'_1(\mu) = b^{-1} \mu^{\frac{1}{b}-1}$ and $\hat{E}'(\mu) = 1$, the linear combination in $\hat{F}(m)$ is chosen in such way 
	that $\hat{F}'(\mu) = \mu^{\frac{1}{b}-1}-{k^{\frac{1}{b}-1}}$ decays to zero at infinity exponentially fast. Then, the constraint (\ref{equn}) 
	is well-defined for every $\tilde{m}\in H^1(\mathbb{R})$ and expresses fixed constraint of $\hat{F}(m)$ up to the first order. 
\end{proof}

\begin{remark}
	\label{rem-equiv}
	It follows from \eqref{ag} and \eqref{muphi} that 
	\begin{equation}
	\label{mu-phi-relation}
	\mu = k \left( \frac{c-k}{c-\phi} \right)^b.
	\end{equation}
As a result, the constraint (\ref{equn}) can be equivalently written as 
\begin{equation}
\label{eqnu-equivalent}
	\langle  (c-\phi)^{b-1}-(c-k)^{b-1}, \tilde{m} \rangle = 0,
\end{equation}
where $(c-\phi)^{b-1}-(c-k)^{b-1}$ decays to zero at infinity exponentially fast.
\end{remark}

\begin{remark}
The linear combination of conserved quantities $\hat{E}$, $\hat{F}_1$ and $\hat{F}_2$ in the construction of $\hat{F}$ is not uniquely defined. However, first variations of conserved quantities must be well defined due to nonzero boundary conditions by Remark \ref{remark-var}. This implies that other well-defined linear combinations of conserved quantities give the constraint which is proportional to the constraint in (\ref{equn}).
\end{remark}

\section{Stability criterion}
\label{sec-stab}

The results of Lemma \ref{lem-EL} suggest that the smooth solitary wave of Lemma \ref{lem-trav} with the profile $\mu = \phi - \phi''$ is a critical point of the action functional $\Lambda_{\omega_1,\omega_2}$ with uniquely selected parameters $(\omega_1,\omega_2)$. In order to prove Theorem \ref{theorem-main}, 
we need to derive the criterion for the critical point to be a local non-degenerate minimizer of $\Lambda_{\omega_1,\omega_2}$ subject to the fixed 
value of another conserved quantity $\hat{F}$. This is done by using the second derivative test, which relies on Corollaries \ref{cor-second-var} and  \ref{cor-constraint}. We will show that the Hessian operator $\mathcal{L}$ 
of the action functional $\Lambda_{\omega_1,\omega_2}$ defined by  (\ref{expansion-lambda}) and (\ref{second-var-1}) has exactly one simple negative eigenvalue and a simple zero eigenvalue isolated from the rest of the spectrum. Then, we add the constraint (\ref{equn}) in order to derive a precise condition when the Hessian operator is positive under the constraint with the only degeneracy due to the translational symmetry. 

The following lemma gives the spectral properties of the linear operator $\mathcal{L}$. 

\begin{lemma}
	\label{lem-operator}
The linear operator $\mathcal{L}$ defined by (\ref{second-var-1}) is extended as a self-adjoint operator in $L^2(\mathbb{R})$ with the dense domain 	$H^2(\mathbb{R}) \subset L^2(\mathbb{R})$. There exists $\delta > 0$ 
such that the spectrum of $\mathcal{L}$ in $(-\infty,\delta)$ consists of a simple zero eigenvalue and a simple negative eigenvalue. 
\end{lemma}

\begin{proof}
The linear operator $\mathcal{L}$ defined by (\ref{second-var-1}) belongs to the class of self-adjoint Sturm--Liouville operators in $L^2(\mathbb{R})$ with the dense domain $H^2(\mathbb{R})$. Since coefficients of $\mathcal{L}$ are smooth and bounded, see the proof of Corollary \ref{cor-second-var}, standard properties of the Sturm--Liouville operators hold.  	

Since $\mu(x) \to k$ as $|x| \to \infty$ exponentially fast, Weyl's Lemma states that the essential spectrum of $\mathcal{L}$ is given by the essential spectrum of the linear operator with constant coefficients $\mathcal{L}_{\infty}$ given by 
\begin{eqnarray}
\mathcal{L}_{\infty} = -\frac{c-k}{k^2} \frac{d^2}{dx^2} 
+ \frac{c-(b+1)k}{k^2}.
\label{second-var-2}
\end{eqnarray}
Since $c > (b+1)k$, the spectrum of $\mathcal{L}_{\infty}$ 
coincides with 
$$
\Sigma_{\infty} := [k^{-2} (c-(b+1)k),\infty),
$$ 
which is strictly positive.  
Hence, the spectrum of the self-adjoint operator $\mathcal{L}$ in 
$\mathbb{R}\backslash \Sigma_{\infty}$ 
consists of isolated semi-simple eigenvalues of finite multiplicity. Moreover, each eigenvalue is simple because the Wronskian $W(v_1,v_2)$ of two solutions $v_1$, $v_2$ of the second-order differential equation $\mathcal{L} v = \lambda v$ satisfies the Liouville formula, 
$$
W(v_1,v_2) = \frac{W_0 \mu^2}{c - \phi},
$$
where $W_0 \neq 0$ is constant so that if $v_1 \in H^2(\mathbb{R})$ 
decays to zero as $|x| \to \infty$, then $v_2 \notin H^2(\mathbb{R})$ diverges to infinity as $|x| \to \infty$.

It remains to characterize the zero and negative eigenvalues of $\mathcal{L}$. 
Due to the translation symmetry of the $b$-CH equation (\ref{bCHm}), 
$\mu' \in H^2(\mathbb{R})$ belongs to the kernel of $\mathcal{L}$ 
so that $0$ is in the spectrum of $\mathcal{L}$. Sturm's Oscillation Theorem 
states that the $n$-th simple eigenvalue corresponds to the eigenfunction with $(n-1)$ simple zeros on $\mathbb{R}$. Since 
$\mu' = \frac{b\mu}{c-\phi} \phi'$ has only one zero on $\mathbb{R}$, $0$ is the second eigenvalue of $\mathcal{L}$ and there exists only one simple negative eigenvalue.
\end{proof} 

In order to prove the next lemma, we first derive formal relations on the family of solutions in the differential equations (\ref{second-order}) and (\ref{first-order}). Solutions depend on three parameters $g$, $a$, and $c$, which are considered independently for the moment.
Formal differentiating of $(c-\phi)^b \mu = a$ in $g$, $a$, and $c$ yields
\begin{equation}
\label{parder}
\partial_g \mu = \frac{b \mu \partial_g \phi}{c-\phi}, \qquad 
\partial_a \mu = \frac{b \mu \partial_a \phi}{c-\phi} + \frac{\mu}{a}, \quad 
\partial_c \mu = \frac{b \mu \partial_c \phi}{c-\phi} - \frac{b\mu}{c-\phi}.
\end{equation}
We note the following result due to the scaling transformation.

\begin{proposition}
	\label{prop-symm-1}
Let $\mu$ be defined by (\ref{muphi}) and assume that $\mu$ is smooth 
with respect to parameters $g$, $a$, and $c$. Then, it satisfies
	\begin{equation}
\label{scaling2}
c \partial_c \mu + (b+1) a \partial_a \mu + 2 g \partial_g \mu = \mu.
\end{equation}
\end{proposition}
	
\begin{proof}	
	Solutions $\phi(x;g,a,c)$ of the system of differential equations (\ref{second-order}) and (\ref{first-order}) enjoy the scaling transformation:
	\begin{equation}
	\label{scaling}
	\phi(x;g,a,c) = c \varphi(x;\gamma,\alpha), \quad g = c^2 \gamma, \quad a = c^{b+1} \alpha,
	\end{equation}
	where $\varphi$, $\gamma$, and $\alpha$ are independent of $c$. It follows from (\ref{scaling}) for a smooth $\phi(x;g,a,c)$ that 
	\begin{equation}
	\label{scaling1}
	c \partial_c \phi + (b+1) a \partial_a \phi + 2 g \partial_g \phi = \phi.
	\end{equation}	
The relation (\ref{scaling2}) follows from (\ref{scaling1}) by using (\ref{parder}).
\end{proof}

Long but straightfoward computations based on the explicit expression (\ref{second-var-1}) with 
$$
(c-k)[(b-1)c + (b+1)k] = 2g + (b-1)c^2
$$ 
yield
\begin{equation}
\label{range-1}
\mathcal{L} \partial_g \mu = \frac{b}{a{(1-b)}} (c-\phi)^{b-1},
\end{equation}
\begin{equation}
\label{range-2}
\mathcal{L} \partial_a \mu = \frac{2g + (b-1)c^2}{a^2 (b-1)} (c-\phi)^{b-1} - \frac{b}{a(b-1)},
\end{equation}
and 
\begin{equation}
\label{range-3}
\mathcal{L} \partial_c \mu = -\frac{bc}{a} (c-\phi)^{b-1},
\end{equation}
where $b \neq 1$ and $a \neq 0$ are assumed. It follows from (\ref{scaling2}) with the help of (\ref{range-1}), (\ref{range-2}), and (\ref{range-3}) that 
\begin{equation}
\label{range}
\mathcal{L}\mu = \frac{2g + (b-1) c^2}{a(b-1)} (c-\phi)^{b-1} - \frac{b(b+1)}{b-1}.
\end{equation}

We also acknowledge the obvious translation symmetry of the $b$-CH equation (\ref{bCHm}).
\begin{proposition}
	\label{prop-symm-2}
	Let $\mu(t,x)$ be a solution of the $b$-CH equation (\ref{bCHm}). Then 
$\mu(t+t_0,x+x_0)$ is also a solution of the $b$-CH equation (\ref{bCHm}) for every $t_0 \in \mathbb{R}$ and $x_0 \in \mathbb{R}$.
\end{proposition}

\begin{proof}
	The proof is immediate since the $b$-CH equation (\ref{bCHm}) has constant coefficients in $(t,x)$.
\end{proof}

We are now ready to prove the following lemma which gives the precise condition for the smooth solitary wave with the profile $\mu$ to become a constrained minimizer of the action functional $\Lambda_{\omega_1,\omega_2}$ given by (\ref{action-1}) and (\ref{multipliers-relation}).

\begin{lemma}
	\label{lem-minimizer}
	The solitary wave with profile $\mu$ is a local constrained minimizer of $\Lambda_{\omega_1,\omega}$ given by (\ref{action-1}) and (\ref{multipliers-relation}) under the constraint (\ref{equn}) if and only if the mapping 
\begin{equation}
\label{map}
	k \mapsto Q(\phi) := \int_{\mathbb{R}} 
	\left[ b \left( \frac{c-k}{c-\phi} \right) - \left(\frac{c-k}{c-\phi} \right)^{b} -b+1 \right]   dx
\end{equation}
is increasing. The minimizer is only degenerate due to the translational symmetry of the $b$-CH equation (\ref{bCHm}) if and only if the mapping (\ref{map}) is strictly increasing.
\end{lemma} 

\begin{proof}
Under the spectral properties of Lemma \ref{lem-operator}, 
it is well-known since the pionering work \cite{VK} that 
the operator $\mathcal{L}$ is positive definite on a subspace 
of $L^2(\mathbb{R})$ defined by the scalar constraint 
\eqref{eqnu-equivalent} if and only if 
$\langle \mathcal{L}^{-1} v_0, v_0 \rangle \leq 0$, where
\begin{equation}
\label{condn}
v_0 := (c-\phi)^{b-1}-(c-k)^{b-1}.
\end{equation}
Moreover, the operator $\mathcal{L}$ is strictly positive definite in 
\begin{equation}
\label{constrained-L2}
L^2_c := \{ \tilde{m} \in L^2(\mathbb{R}) : \quad \langle v_0, \tilde{m} \rangle = 0, \quad \langle \mu', \tilde{m}\rangle = 0 \}
\end{equation}
if and only if $\langle \mathcal{L}^{-1} v_0, v_0 \rangle < 0$. 
It remains to show that the sign of $\langle \mathcal{L}^{-1} v_0, v_0 \rangle$ 
is opposite to the sign of the derivative of the mapping (\ref{map}).

Let us recall that parameters $g$ and $a$ are expressed in terms 
of $k$ for fixed $b > 1$ and $c > 0$ from (\ref{ag}) and that 
the family of solitary waves is smooth with respect to parameter $k$. Differentiating $\mu$
with respect to $k$ yields with the chain rule that
\eq{
\partial_k \mu = (c - (b+1)k) \left[ \partial_g \mu + (c-k)^{b-1} \partial_a \mu \right].
}{muk}
From \eqref{muk}, we use (\ref{range-1}), (\ref{range-2}), and 
(\ref{range}) to obtain
\eq{
\mathcal{L}\partial_k \mu = \frac{c - k(b+1)}{ak (b-1)} \left[ (k + c(b-1)) (c-\phi)^{b-1} - bk (c-k)^{b-1} \right],
}{Lmuk}
and
\eq{
\mathcal{L} \mu = \frac{c-k}{a(b-1)} \left[(k(b+1) + c(b-1)) (c-\phi)^{b-1} -bk(b+1)(c-k)^{b-1} \right].
}{Lmu}
Since $\mu(x) \to k$ as $|x| \to \infty$ exponentially fast, then $\partial_k \mu(x) \to 1$ as $|x| \to \infty$, so that $\partial_k \mu(x)$ does not decay to $0$ at infinity. However, $\partial_k \mu(x) - k^{-1} \mu(x)$ does decay to $0$ at infinity, so we use \eqref{Lmuk} and \eqref{Lmu} to compute  
	$$
	\mathcal{L}\left( \partial_k \mu - k^{-1} \mu\right) = 
	-\frac{c b^2}{a(b-1)} \left[ (c-\phi)^{b-1} - (c-k)^{b-1} \right].
	$$
Therefore, 
\eq{\mathcal{L}^{-1} v_0 = -\frac{a(b-1)}{c b^2} (\partial_k \mu - k^{-1} \mu),}{Lmv} 
so that we can compute 
\begin{align*}
\langle \mathcal{L}^{-1} v_0, v_0 \rangle &= -\frac{a(b-1)}{c b^2} \int_{\mathbb{R}} 
\left[ (c-\phi)^{b-1} - (c-k)^{b-1} \right] \left[ \partial_k \mu - \frac{\mu}{k} \right] dx,
\end{align*}
which becomes in view of relation (\ref{mu-phi-relation}), 
\begin{align*}
\langle \mathcal{L}^{-1} v_0, v_0 \rangle &= -\frac{a k(b-1)}{c b} (c-k)^{b-1}  \int_{\mathbb{R}} 
\left[ 1 - \left(\frac{c-k}{c-\phi} \right)^{b-1} \right] \frac{\partial}{\partial k} \left( \frac{c-k}{c-\phi} \right) dx \\
&= -\frac{a k(b-1)}{c b^2} (c-k)^{b-1}  \frac{\partial}{\partial k} \int_{\mathbb{R}} 
\left[ b \left( \frac{c-k}{c-\phi} \right) - \left(\frac{c-k}{c-\phi} \right)^{b} -b+1 \right]   dx,
\end{align*}
where the integrands has been normalized to converge to zero at infinity. 
Since $a > 0$, $b > 1$, and $c > 0$, the sign of $\langle \mathcal{L}^{-1} v_0, v_0 \rangle$ is opposite to the sign of the derivative of the mapping (\ref{map}).
\end{proof}

\begin{remark}
The result of Theorem \ref{theorem-main} follows from 
the second derivative test in Lemma \ref{lem-minimizer}, 
the local well-posedness theory for the $b$-CH equation 
(\ref{bCHm}) in $X_k$, and the orbital stability theory 
pioneered in \cite{GSS}. 
\end{remark}

\section{Verification of the stability criterion}
\label{sec-num}

Here we verify the stability criterion of Theorem \ref{theorem-main}. 
To do so, we perform some transformations and rewrite 
the second-order equation (\ref{CHode}) with $g = kc - \frac{1}{2} (b+1)k^2$ 
as 
\begin{equation}
\label{CH-1}
(c-\phi) (\phi - \phi'') + \frac{1}{2} (b-1) (\phi'^2 - \phi^2)= c k - \frac{1}{2} (b+1) k^2.
\end{equation}
By using the transformation 
(which generalizes the one used in \cite{CS-02} and \cite{GeyerSAPM} for $b = 2$),
\begin{equation}
\label{CH-2}
z = \int_0^x \frac{dx}{[c - \phi(x)]^{\frac{b-1}{2}}}, \quad \phi(x) = \psi(z),
\end{equation}
we obtain from (\ref{CH-1}) the equivalent second-order equation,
\begin{equation}
\label{CH-3}
- \psi''(z) + (\psi - k) (c- \psi)^{b-2} \left[ c - \frac{1}{2} (b+1) (\psi + k) \right] = 0.
\end{equation}
Using the following transformation 
\begin{equation}
\label{CH-4}
\zeta = \sqrt{c - k(b+1)} (c - k)^{\frac{b-2}{2}} z, \quad \psi(z) = k + (c-k) \varphi(\zeta),
\end{equation}
we rewrite the second-order equation (\ref{CH-3}) in the normalized form 
\begin{equation}
\label{varphi-eq}
- \varphi''(\zeta) +  \varphi (1 - \varphi)^{b-2} \left[ 1 - (2\gamma)^{-1} (b+1) \varphi \right] = 0,
\end{equation}
where 
\begin{equation}
\label{gamma}
\gamma := \frac{c - k(b+1)}{c - k}
\end{equation}
is the normalized parameter. It follows that $\gamma \in (0,1)$ if $k \in (0,(b+1)^{-1} c)$. Substituting (\ref{CH-2}) and (\ref{CH-4}) into $Q(\phi)$ 
given by (\ref{map-intro}) yields
\begin{align}
\nonumber
Q(\phi) &= \int_{\mathbb{R}} \left[ b \frac{\phi - k}{c - \phi} 
+ 1 - \left( \frac{c-k}{c-\phi} \right)^b \right] dx \\
\nonumber
&= \int_{\mathbb{R}} \left[ b (\psi - k) (c-\psi)^{\frac{b-3}{2}}
+ (c-\psi)^{\frac{b-1}{2}} -  (c-k)^b (c-\psi)^{-\frac{b+1}{2}} \right] dz  \\
&= \gamma^{-1/2} \int_{\mathbb{R}}  \left[ b \varphi (1 - \varphi)^{\frac{b-3}{2}}
+ (1 - \varphi)^{\frac{b-1}{2}} -  (1 - \varphi)^{-\frac{b+1}{2}} \right] d\zeta.
\label{Q-form}
\end{align}
Although we do not write it explicitly, the wave profile $\varphi$ depends on $\gamma$ since the differential equation (\ref{varphi-eq}) for $\varphi$ depends on $\gamma$. If $b > 1$ and $c > 0$ are fixed, monotonicity of 
the mapping $k \mapsto Q(\phi)$ given by (\ref{map-intro}) is determined 
with the chain rule from 
monotonicity of the mappings $k \mapsto \gamma$ and $\gamma \mapsto Q(\phi)$ given by (\ref{gamma}) and (\ref{Q-form}). Since
$$
\frac{d \gamma}{d k} = \frac{-bc}{(c-k)^2} < 0, 
$$
the mapping (\ref{map-intro}) is strictly increasing if and only if 
\begin{equation}
\label{stab-criterion}
\frac{d}{d\gamma} Q(\phi) < 0.
\end{equation}

The following two lemmas report explicit computations of $Q(\phi)$ 
for the integrable cases $b = 2$ and $b = 3$, from which the stability criterion (\ref{stab-criterion}) can be proven analytically. 

\begin{lemma}
	\label{lemma-CH}
	The mapping (\ref{map-intro}) is strictly increasing for $b = 2$, 
	$c > 0$, and $k \in \left(0,\frac{1}{3} c \right)$.
\end{lemma}

\begin{proof}
The second-order equation (\ref{varphi-eq}) with $b = 2$ admits the exact solution 
for the solitary wave centered at $\zeta = 0$:
\begin{equation}
	\label{KdV-sol}
\varphi(\zeta) = \gamma \; {\rm sech}^2 \big(  \frac{1}{2} \zeta \big).
\end{equation}
Substituting (\ref{KdV-sol}) to $Q(\phi)$ in (\ref{Q-form}) for $b = 2$ yields
\begin{align*}
	Q(\phi)	&= -\gamma^{-1/2} \int_{\mathbb{R}} (1-\varphi)^{-3/2} \varphi^2 d\zeta = -\gamma^{3/2} \int_{\mathbb{R}}  \frac{ {\rm sech}^4(\frac{1}{2} \zeta)}{(1 - \gamma {\rm sech}^2(\frac{1}{2} \zeta))^{3/2}} d\zeta,
	\end{align*}
from which we compute
	\begin{align*}
	\frac{d}{d\gamma} Q(\phi) = -\frac{3}{2} \gamma^{1/2} \int_{\mathbb{R}} \frac{{\rm sech}^4( \frac{1}{2} \zeta)}{(1 - \gamma {\rm sech}^2( \frac{1}{2} \zeta))^{5/2}} d\zeta < 0,
	\end{align*}
	hence the stability criterion (\ref{stab-criterion}) is satisfied.
\end{proof}

\begin{lemma}
	\label{lemma-DP}
	The mapping (\ref{map-intro}) is strictly increasing for $b = 3$, 
	$c > 0$, and $k \in \left(0,\frac{1}{4} c \right)$.
\end{lemma}

\begin{proof}
The second-order equation (\ref{varphi-eq}) with $b = 3$ admits the exact solution 
for the solitary wave centered at $\zeta = 0$:
	\begin{equation}
	\label{Gardner-sol}
\varphi(\zeta) = \frac{3 \gamma}{2 + \gamma + \sqrt{(1-\gamma) (4 - \gamma)} \cosh(\zeta)}.
	\end{equation}
Substituting (\ref{Gardner-sol}) to $Q(\phi)$ in (\ref{Q-form}) for $b = 3$ yields
	\begin{align*}
	Q(\phi) &= -\gamma^{-1/2} \int_{\mathbb{R}} \frac{\varphi^2 (3 - 2 \varphi)}{(1-\varphi)^2} d\zeta.
	\end{align*}
Writing $\partial_{\gamma} \varphi = \gamma^{-1} \varphi + \hat{\varphi}$ 
with 
$$
\hat{\varphi}(\zeta) = \frac{3 \gamma \;
[ (5 - 2 \gamma) \cosh(\zeta) - 2  \sqrt{(1-\gamma) (4 - \gamma)}]}{2  \sqrt{(1-\gamma) (4 - \gamma)} (2 + \gamma +  \sqrt{(1-\gamma) (4 - \gamma)} {\rm sech}(\zeta))^2},
$$	
we obtain explicitly
\begin{align*}
\frac{d}{d\gamma} Q(\phi) &= -\frac{1}{2 \gamma^{3/2}} 
\int_{\mathbb{R}} \frac{\varphi^2 (9 - 7 \varphi + 2 \varphi^2)}{(1-\varphi)^3} d\zeta -\frac{2}{\gamma^{1/2}} 
\int_{\mathbb{R}} \frac{\varphi (3 - 3 \varphi + \varphi^2)}{(1-\varphi)^3} \hat{\varphi} d\zeta.
\end{align*}
Both terms are strictly negative. Indeed, the first term is negative because
$$
\min_{\varphi \in [0,1]} (9 - 7 \varphi + 2 \varphi^2) = 4 > 0.
$$ 
The second term is negative if $\hat{\varphi}(\zeta) > 0$ for every $\zeta \in \mathbb{R}$, which is true for every $\gamma \in (0,1)$ because 
$$
(5 - 2 \gamma) \cosh(\zeta) - 2  \sqrt{(1-\gamma) (4 - \gamma)} 
\geq (5 - 2 \gamma) - 2  \sqrt{(1-\gamma) (4 - \gamma)} > 0.
$$
Hence the stability criterion (\ref{stab-criterion}) is satisfied.
\end{proof}

In the general case $b > 1$, we verify the stability criterion (\ref{stab-criterion}) numerically. For this, we integrate the second-order 
equation (\ref{varphi-eq}) with the boundary conditions $\varphi(\zeta) \to 0$ as $|\zeta| \to \infty$ and obtain the first-order invariant:
\begin{equation}
\label{varphi-invariant}
-(\varphi')^2 + \frac{(1-\varphi)^{b-1}}{\gamma b (b-1)} \left[ 2(1-\gamma) + 2(1-\gamma) (b-1) \varphi + b (b-1) \varphi^2 \right] = \frac{2(1-\gamma)}{\gamma b (b-1)}.
\end{equation}
If the center of the solitary wave is translated to the origin, 
so that $\varphi'(0) = 0$, we can parameterize the integral for $Q(\phi)$ in (\ref{Q-form}) and obtain the equivalent representation
\begin{equation}
\label{Q-equiv}
Q(\phi) = \int_0^{\varphi_0} 
\frac{ 2 \sqrt{b(b-1)} (1-\varphi)^{-\frac{b+1}{2}} \left[ (1-\varphi)^{b-1} [(b-1) \varphi + 1]  - 1 \right]d\varphi}{\sqrt{(1-\varphi)^{b-1} \left[ 2(1-\gamma) + 2(1-\gamma) (b-1) \varphi + b (b-1) \varphi^2 \right] - 2(1-\gamma)}},
\end{equation}
where $\varphi_0$ is the turning point, for which the integral is weakly singular due to vanishing denominator. 

Figure \ref{fig-momentum} shows dependence of $Q(\phi)$ computed from the numerical quadrature of the integral (\ref{Q-equiv}) versus $\gamma$ in $[0,1]$
for different integer values of $b$. The turning point $\varphi_0$ was obtained by using the Newton--Raphson iterative method. The integrals in the expression (\ref{Q-equiv}) were computed by using the composite midpoint rules. It follows from the behavior of $Q(\phi)$ versus $\gamma$ that the stability 
criterion (\ref{stab-criterion}) is satisfied for every $b > 1$.

\begin{figure}[htb!]
	\includegraphics[width=0.6\textwidth]{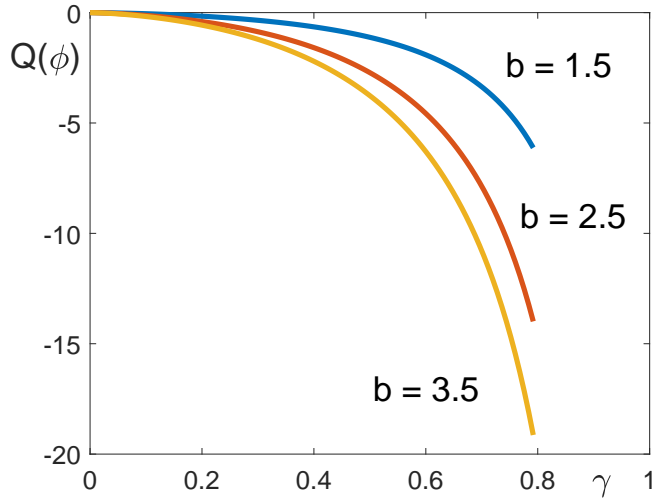}
	\caption{$Q(\phi)$ versus $\gamma$ for three values of $b$} 
	\label{fig-momentum}
\end{figure}

The final lemma shows that the stability criterion (\ref{stab-criterion}) 
is satisfied in the asymptotic limits $\gamma \to 0$ and $\gamma \to 1$ for every $b > 1$. In view of the transformation (\ref{gamma}), 
they corresponds to the asymptotic limits $k \to 0$ and $k \to (b+1)^{-1} c$ for fixed $b > 1$ and $c > 0$. 

\begin{lemma}
	\label{lemma-asymptotics}
	The mapping (\ref{map-intro}) is strictly increasing for $b > 1$ and $c > 0$ in the asymptotic limits $k \to 0$ and $k \to (b+1)^{-1} c$. 
\end{lemma}

\begin{proof}
It follows from (\ref{varphi-eq}) that $\varphi$ depends on $\gamma$ analytically and $\varphi \to 0$ as $\gamma \to 0$. Therefore, it satisfies the Taylor expansion 
\begin{equation*}
\varphi(\zeta) = \gamma \varphi_1(\zeta) + \gamma^2 \varphi_2(\zeta) + \mathcal{O}(\gamma^3) \quad \mbox{\rm as} \;\; \gamma \to 0,
\end{equation*}
where $\varphi_1$ is computed similarly to (\ref{KdV-sol}) in the form 
$$
\varphi_1(\zeta) = \frac{3}{b+1} {\rm sech}^2(\frac{1}{2} \zeta).
$$
Similarly to the proof of Lemma \ref{lemma-DP}, we write $\partial_{\gamma} \varphi = \gamma^{-1} \varphi + \hat{\varphi}$, where 
$\hat{\varphi} = \gamma \varphi_2+ \mathcal{O}(\gamma^2)$, and obtain from $Q(\phi)$ in (\ref{Q-form}) that 
\begin{align*}
\frac{d}{d\gamma} Q(\phi) &= -\frac{1}{2 \gamma^{3/2}} 
\int_{\mathbb{R}} \left[ ((b+2) \varphi - 1) (1-\varphi)^{-\frac{b+3}{2}} + (1 - 3 \varphi + (b-1)(b-2) \varphi^2) (1 - \varphi)^{\frac{b-5}{2}} \right] d\zeta \\
& -\frac{1}{2 \gamma^{1/2}} 
\int_{\mathbb{R}} \left[ (b+1)  (1-\varphi)^{-\frac{b+3}{2}} + ((b-1)^2 \varphi - (b+1)) (1 - \varphi)^{\frac{b-5}{2}} \right] \hat{\varphi} d\zeta.
\end{align*}
We note that 
\begin{align*}
f(\varphi) &:= ((b+2) \varphi - 1)  + (1 - 3 \varphi + (b-1)(b-2) \varphi^2) (1 - \varphi)^{b-1} \\
&= \frac{3}{2} b (b-1) \varphi^2 + \mathcal{O}(\varphi^3)  \quad \mbox{\rm as} \;\; \varphi \to 0
\end{align*}
and 
\begin{align*}
g(\varphi) &:= (b+1) + ((b-1)^2 \varphi - (b+1)) (1 - \varphi)^{b-1}  \\
&= 2 b (b-1) \varphi + \mathcal{O}(\varphi^2)  \quad \mbox{\rm as} \;\; \varphi \to 0.
\end{align*}
By substituting the asymptotic expansions with 
$\varphi = \mathcal{O}(\gamma)$ and $\hat{\varphi} = \mathcal{O}(\gamma)$, 
we obtain 
\begin{align*}
\frac{d}{d\gamma} Q(\phi) = -\frac{3}{4} b (b-1) \gamma^{1/2} \int_{\mathbb{R}} \varphi_1^2 d\zeta + \mathcal{O}(\gamma^{3/2}) \quad \mbox{\rm as} \;\; \gamma \to 0,
\end{align*}
so that the stability criterion (\ref{stab-criterion}) is satisfied as $\gamma \to 0$ for every $b > 1$.

In the opposite asymptotic limit $\gamma \to 1$, it follows from (\ref{Q-equiv}) that $\varphi_0 \to 1$ as $\gamma \to 1$ for $b > 1$ 
and 
$$
\lim_{\gamma \to 1} Q(\phi) = 2 \int_0^1 \frac{(1-\varphi)^{b-1} [(b-1) \varphi + 1] -1}{\varphi (1-\varphi)^b} d \varphi = -\infty,
$$
so that the stability criterion (\ref{stab-criterion}) is satisfied as $\gamma \to 1$ for every $b > 1$.
\end{proof}

\section{Conclusion}
\label{sec-conclude}

We have derived the precise criterion for orbital stability of smooth solitary waves on the nonzero constant background in the $b$-CH equation (\ref{bCH}) with $b > 1$. Perturbations to the horizontal velocity $u(t,x)$ are controlled in $H^3(\mathbb{R})$. Verification of this stability criterion analytically for integrable cases $b = 2$ and $b = 3$ give alternative proofs of orbital stability of smooth solitary waves in the Camassa--Holm and Degasperis--Procesi equations compared to the work in \cite{CS-02} and \cite{Liu-21} respectively. We also verified the stability criterion analytically for every $b > 1$ and $c > 0$ in the asymptotic limits, where the family of smooth solitary waves terminates. The stability criterion is verified numerically in the general case. It is still open to verify the stability criterion analytically for every $b > 1$, $c > 0$, and $k \in (0,(b+1)^{-1} c)$.

This work opens roads to further studies of travelling waves in the $b$-CH model. Stability of smooth periodic waves for the integrable case $b = 2$ 
was considered in \cite{GeyerSAPM} by using two alternative Hamiltonian structure, both are different from the Hamiltonian structure (\ref{sympl-1}) used here. This approach also leads to the precise stability criterion which can only be verified numerically in the general case. It would be natural to explore the Hamiltonian structure (\ref{sympl-1}) for stability of smooth periodic waves both for $b = 2$ and generally for $b > 1$. 

Stability of smooth travelling waves in the $b$-CH model with $b \leq 1$ 
is also of interest from the points of physical applications. Some smooth traveling solitary waves with positive $\phi$ and $\mu$ exist for $b \leq 1$ and their orbital stability can be clarified by using the same analysis as in the proof of Theorem \ref{theorem-main}. Orbital stability of periodic waves for $b \leq 1$ is also an open problem.

\vspace{0.25cm}

{\bf Acknowledgement.} 
The research of S. Lafortune was supported by a Collaboration Grants for Mathematicians from the Simons Foundation (award \# 420847). 
The research of D. E. Pelinovsky was supported by the NSERC Discovery grant. 
The authors thank R. M. Chen, D. Holm, and A. Hone for useful comments 
on the earlier stage of this work.

\bibliographystyle{unsrt}

\end{document}